\documentclass[submission,copyright,creativecommons]{eptcs}
\providecommand{\event}{ICLP DC 2021} 

\usepackage{graphicx}
\usepackage{listings}
\usepackage{booktabs}
\usepackage{subcaption}
\usepackage{algorithm}
\usepackage{amsthm}
\newtheorem{definition}{Definition}

\newtheorem{corollary}{Corollary}
\usepackage{algpseudocode}

\title{Graph Based Answer Set Programming Solver Systems}
\author{Fang Li
\institute{University of Texas at Dallas\\ Richardson, USA}
\email{fang.li@utdallas.edu}
}
\def\titlerunning{Graph Based Answer Set Programming Solver Systems}
\def\authorrunning{F. Li}
\begin{document}
\maketitle

\begin{abstract}
Answer set programming (ASP) is a popular nonmonotonic-logic based paradigm for knowledge representation and solving combinatorial problems. Computing the answer set of an ASP program is NP-hard in general, and researchers have been investing significant effort to speed it up. The majority of current ASP solvers employ SAT solver-like technology to find these answer sets. As a result, justification for why a literal is in the answer set is hard to produce. There are dependency graph based approaches to find answer sets, but due to the representational limitations of dependency graphs, such approaches are limited. This paper proposes a novel dependency graph-based approach for finding answer sets in which conjunction of goals is explicitly represented as a node which allows arbitrary answer set programs to be uniformly represented. Our representation preserves causal relationships allowing for justification for each literal in the answer set to be elegantly found. In this paper, we explore two different approaches based on the graph representation: bottom-up and top-down. The bottom-up approach finds models by assigning truth values along with the topological order, while the top-down approach generates models starting from the constraints.
\end{abstract}

\section{Introduction}
Answer set programming (ASP) \cite{GL2,MT5,EFPL12} is a popular nonmonotonic-logic based paradigm for  knowledge representation and solving combinatorial problems. Computing the answer set of an ASP program is NP-hard in general, and researchers have been investing significant effort to speed it up. Most ASP solvers employ SAT solver-like technology to find these answer sets. As a result, justification for why a literal is in the answer set is hard to produce. There are dependency graph (DG) based approaches to find answer sets, but due to the representational limitations of dependency graphs, such approaches are limited. In this paper we propose a novel dependency graph-based approach for finding answer sets in which conjunction of goals is explicitly represented as a node which allows arbitrary answer set programs to be uniformly represented. Our representation preserves causal relationships allowing for justification for each literal in the answer set to be elegantly found.

Compared to SAT solver based implementations, graph-based implementations of ASP have not been well studied. Very few researchers have investigated graph-based techniques. NoMoRe system \cite{anger2001nomore} represents ASP programs with a \textit{block graph} (a labeled graph) with meta-information, then computes the A-coloring (non-standard graph coloring with two colors) of that graph to obtain answer sets. Another approach \cite{konczak2005graphs} uses \textit{rule dependency graph} (nodes for rules, edges for rule dependencies) to represent ASP programs, then performs graph coloring algorithm to determine which rule should be chosen to generate answer sets. Another group \cite{linke2005suitable} proposed an hybrid approach which combines different kinds of graph representations that are suitable for ASP. The hybrid graph uses both rules and literals as nodes, while edges represent dependencies. It also uses the A-coloring technique to find answer sets.

All of the above approaches were well designed, but their graph representations are complex as they all rely on extra information to map the ASP elements to nodes and edges of a graph. In contrast, our approach uses a much simpler graph representation, where nodes represent literals and an edge represent the relationship between the nodes it connects. Since this representation faithfully reflects the causal relationships, it is capable of producing causal justification for goals entailed by the program.

In the rest of the paper, we assume that the reader is familiar with ASP. Details can be found elsewhere \cite{baral,gelfond2014knowledge}.

\section{Dependency Graph}\label{sec:dg}

A dependency graph \cite{linke2005suitable} uses nodes and directed edges to represent dependency relationships of an ASP rule. 

\begin{definition}
The dependency graph of a program is defined on its literals s.t. there is a positive (resp. negative) edge from $p$ to $q$ if $p$ appears positively (resp. negatively) in the body of a rule with head $q$.
\label{def1}
\end{definition}

Conventional dependency graphs are not able to represent ASP programs uniquely. This is due to the inability of dependency graphs to distinguish between non-determinism (multiple rules defining a proposition) and conjunctions (multiple conjunctive sub-goals in the body of a rule) in logic programs. For example, the following two programs have identical dependency graphs (Figure \ref{fig:fig1}).\\
\noindent\begin{minipage}{.45\textwidth}
\begin{lstlisting}[language=prolog,basicstyle=\small]
    %% program 1  
    p :- q, not r, not p.
\end{lstlisting}
\end{minipage}
\begin{minipage}{.45\textwidth}
\begin{lstlisting}[language=prolog,basicstyle=\small]
    %% program 2
    p :- q, not p. p :- not r.
\end{lstlisting}
\end{minipage}

To make conjunctive relationships representable by dependency graphs, we first transform it slightly to come up with a novel representation method. This new representation method, called conjunction node representation (CNR) graph, uses an artificial node to represent conjunction of sub-goals in the body of a rule. This conjunctive node has a directed edge that points to the rule head (Figure \ref{fig:fig2}).

\begin{figure}[tb]
\centering
\begin{minipage}{.5\textwidth}
  \centering
  \includegraphics[scale=0.3]{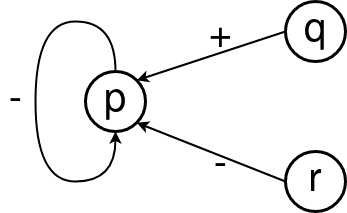}
      \caption{Dep. Graph for Programs 1 \& 2}
    \label{fig:fig1}

\end{minipage}%
\begin{minipage}{.5\textwidth}
  \centering
    \begin{subfigure}[b]{0.5\linewidth}
        \centering
        \includegraphics[scale=0.3]{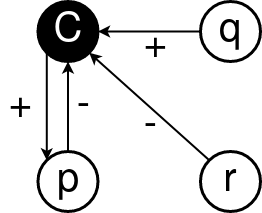}
        \caption{CNR for Program 1}
    \end{subfigure}\hfill
    \begin{subfigure}[b]{0.5\linewidth}
        \centering
        \includegraphics[scale=0.3]{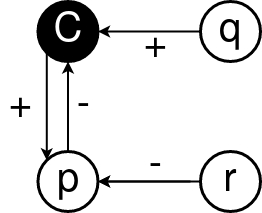}
        \caption{CNR for Program 2}
    \end{subfigure}
    \caption{CNRs for Program 1 \& 2}
    \label{fig:fig2}
\end{minipage}
\end{figure}

The conjunction node, which is colored black, refers to the conjunctive relation between the in-coming edges from nodes representing subogals in the body of a rule. Note that a CNR graph is not a conventional dependency graph.

\medskip\noindent\textbf{Converting CNR Graph to Dependency Graph} \label{sec:cnrtodg}
Since CNR graph does not follow the dependency graph convention, we need to convert it to a proper dependency graph in order to perform dependency graph-based reasoning. We use a simple technique to convert a CNR graph to an equivalent conventional dependency graph. We negate all in-edges and out-edges of the conjunction node. This process essentially converts a conjunction into a disjunction. Once we do that we can treat the conjunction node as a normal node in a dependency graph. As an example, Figure \ref{fig:fig3} shows the CNR graph to dependency graph transformation for program {\tt p :- q, not r.} This transformation is a simple application of De Morgan's law. The rule in this program represents {\tt p :- C. and C :- q, not r.} The transformation produces the equivalent rules {\tt p :- not C., C :- not q. and C :- r.}


\noindent Since conjunction nodes are just helper nodes which allow us to perform dependency graph reasoning, we don't report them in the final answer set.

\medskip\noindent\textbf{Constraint Representation} \label{sec:constraintnode}
ASP also allows for special types of rules called constraints. There are two ways to encode constraints: (i) headed constraint where negated head is called directly or indirectly in the body (e.g., Program 3), and (ii) headless constraints (e.g., Program 4). \\
\noindent\begin{minipage}{.45\textwidth}
\begin{lstlisting}[language=prolog,basicstyle=\small]
    %% program 3
    p :- not q, not r, not p.
\end{lstlisting}
\end{minipage}
\begin{minipage}{.45\textwidth}
\begin{lstlisting}[language=prolog,basicstyle=\small]
    %% program 4
    :- not q, not r.
\end{lstlisting}
\end{minipage}

Our algorithm models these constraint types separately. For the former one, we just need to apply the CNR-DG transformation directly. Note that the head node connects to the conjunction node both with an in-coming edge and an out-going edge (Figure \ref{fig:fig4a}). For the headless constraint, we create a head node with truth value as \textit{False}. 

\begin{figure}[tb]
\centering
\begin{minipage}{.5\textwidth}
  \centering
  \includegraphics[width=.6\linewidth]{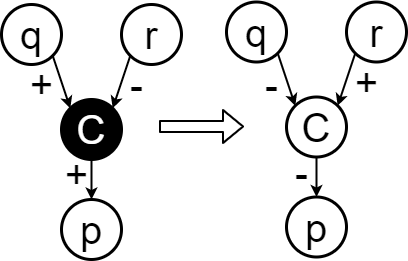}
  \caption{CNR-DG Transformation}
  \label{fig:fig3}
\end{minipage}%
\begin{minipage}{.5\textwidth}
  \centering
    \begin{subfigure}[b]{0.5\linewidth}
        \centering
        \includegraphics[scale=0.3]{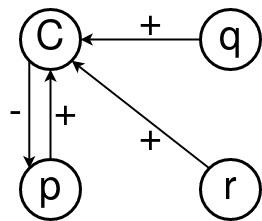}
        \caption{Program 3}
        \label{fig:fig4a}
    \end{subfigure}\hfill
    \begin{subfigure}[b]{0.5\linewidth}
        \centering
        \includegraphics[scale=0.3]{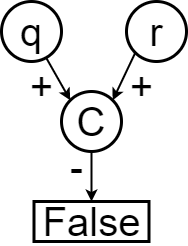}
        \caption{Program 4}
        \label{fig:fig4b}
    \end{subfigure}
    \caption{Constraint DG}
    \label{fig:fig4}
\end{minipage}
\end{figure}

The reason why we don't treat a headless constraint the same way as a headed constraint is because in the latter case, if head node ({\tt p} in Program 3) is provable through another rule, then the headed constraint is inapplicable. Therefore, we cannot simply assign a false value to its head.

\section{grASP: A Bottom-up Approach}
We have developed the grASP graph-based algorithm for finding answer sets. The philosophy of grASP is to translate an ASP program into a dependency graph (via CNR conversion), then propagate truth values from nodes whose values are known to other connected nodes, obeying the sign on the edge, until the values of all the nodes are fixed. However, due to possible existence of a large number of cycles, the propagation process is not straightforward. In grASP, we define a collection of rules for propagating values among nodes involved in cycles. These assignment rules take non-monotonicity of answer set program and the causal relationship among nodes in the dependency graph into account.

Unlike other SAT-solver based approaches, our graph based approach enables stratification of ASP programs on the basis of dependence. The Splitting Theorem \cite{lifschitz1994splitting} can thus be  used to link the various levels, permitting values to be propagated among nodes more efficiently. Also, the existence of sub-structures (sub-graphs) makes an efficient recursive implementation algorithm possible.

\subsection{Cycles in Program} \label{sec:cycle}

In an ASP program, cycles among literals may exist. There are three kinds of cycles that can be found in the program: even cycles, odd cycles, and positive cycles. Even cycles and odd cycles refer to cycles that have even or odd number of negative edges in the corresponding dependency graph. Positive cycles are cycles with no negative edge. It is well known that even cycles generate multiple worlds, while odd cycles kill worlds. For example, in program {\tt p :- not q. q :- not p.}, $p$ and $q$ forms an even cycle, which generates two mutually exclusive worlds: \{p/True, q/False\} and \{q/True, p/False\}. For program {\tt p :- not q. q :- not r. r :- not p.}, nodes $p$, $q$ and $r$ form an odd cycle, which makes the program unsatisfiable. 

\subsection{The grASP Algorithm} \label{sec:graspalgorithm}

The  grASP algorithm is recursive in nature. Since a dependency graph represents the causal relationships among nodes, the reasoning should follow a topological order. We don't need to do topological sorting to obtain the order, instead, for each iteration, we just pick those nodes which have no in-coming edges. We call this kind of node a \textbf{root node}. After picking the root nodes, the algorithm checks their values. If a root node's value has not been fixed (no value yet), we assign \textit{False} to it. Otherwise, the root node will keep its value as is. Once all root nodes' values are fixed, we will propagate the values along their out-going edges in accordance with the sign on each edge (the propagation rules will be discussed in Section \ref{sec:propagate}). At the end of this iteration, we remove all root nodes from the graph, then pass the rest of the graph to the recursive call for the next iteration.

The input graph may contain cycles, and, of course, there will be no root nodes in a cycle. Therefore, this recursive process will leave a cycle unchanged. To cope with this issue, we proposed a novel solution, which wraps all nodes in the same cycles together, and treat the wrapped nodes as a single \textit{virtual} node. All the in-coming and out-going edges connecting the wrapped nodes to other nodes will be incident on or emanate from the virtual node. Thus, the graph is rendered acyclic and ready for the root-finding procedure. 

For each iteration of the recursive procedure, we have to treat regular root nodes and virtual root nodes differently. If the node is a regular node, we do the value assignment, but if it is a virtual node, we will have to \textit{break the cycles}. Cycle breaking means that we will remove the appropriate cycle edges by assigning truth values to the nodes involved (\textit{cycle breaking} will be discussed in this section later). After cycle breaking, we will pass the nodes and edges in this virtual node to another recursive call, because the virtual node can be seen as a substructure of the program. The returned value of the recursive call will be the answer set of the program constituting the virtual node. When all regular and virtual root nodes are processed, we will have to merge the values for propagation.

The value propagation in each iteration makes use of the \textbf{splitting theorem} \cite{lifschitz1994splitting} (details omitted due to lack of space). After removing root nodes,  rest of the graph acts as the \textit{top} strata and all of the predecessors constitute the \textit{bottom} strata, using the terminology of \cite{lifschitz1994splitting}. Thus, when we reach the last node in the topological order, we will get the whole answer set.

The cycle breaking procedure may return multiple results, because a negative even cycle generates two worlds (as discussed in Section \ref{sec:cycle}). Therefore, the merging of solution for the root nodes may possibly result in exponential number of solutions. For example, if the root nodes consists of one regular node and two virtual nodes, each virtual node generates two worlds \& the merging process will return four worlds. Of course, this exponential behavior is inherent to ASP.

\medskip\noindent\textbf{Propagation Rules} \label{sec:propagate}
In an ASP rule, the head term only can be assigned as \textit{True} if all its body term(s) are true. For example, in rule {\tt p :- not q.}, only when {\tt q} is unknown or known as \textit{False}, {\tt p} will be \textit{True}. For another example {\tt p :- q.}, {\tt p} will be \textit{True}, only when {\tt q} is known as \textit{True}. In both examples, {\tt p} will not be assigned as \textit{False}, until the reasoning of the whole program fails to make it \textit{True}. Therefore, mapping this to our graph representation, we obtain two propagation rules: (i) when a node $N$ has a \textit{True} value, assign \textit{True} to all the nodes connected to $N$ via positive out-going edges of $N$; (ii) when a $N$ node has a \textit{False} value, assign \textit{True} to all the nodes connected to $N$ via negative out-going edges of $N$.

\medskip\noindent\textbf{Cycle Wrapping}
As previously mentioned (Section \ref{sec:graspalgorithm}), those nodes which are involved in the same cycles need to be wrapped into a virtual node. The reason being that we want the tangled nodes to act like a single node, in order to be found as a root. This requires the dependency of the wrapped nodes to be properly handled. The virtual node should inherit the dependencies of all the node it contains. These dependency relations include both incoming and outgoing edges.

Since cycles may be overlapped or nested with each other, we can make use of the strongly connected component concept in graph theory. Thus, each strongly connected component will be a virtual node.

\medskip\noindent\textbf{Cycle Breaking} \label{sec:cyclebreaking}
We can state the following corollary:

\begin{corollary}
In the dependency graph of an answer set program, if a node's value is True, all of its in-coming edges and negative out-going edges can be removed. If a node's value is False, then all its positive out-edges can be removed.
\label{corol}
\end{corollary}

\begin{proof}
According to the propagation rules (discussed in Section \ref{sec:propagate}), a node can only be assigned \textit{True} through in-coming edges. When a node has already been known as \textit{True}, it no longer needs any assignment, then all in-coming edges become meaningless. Also, a negative edge won't be able to propagate the \textit{True} value to the other side. If a node has been known as \textit{False}, we still have to keep its in-coming edges to detect inconsistency (if some of its predecessors attempt to assign it as \textit{True}, the program is inconsistent). Note that a node labeled \textit{False} cannot make any node \textit{True} through an outgoing positive edge.
\end{proof}
\noindent We will use Corollary \ref{corol} to remove edges while breaking cycles. The most important thing for cycle breaking is that we need to follow a specific order with respect to types of cycles. As we discussed in Section \ref{sec:cycle}, an even cycle can divide a world into two, while an odd cycle will make one or more worlds unsatisfiable and thus disappear (this happens if no node in the odd cycle has the value \textit{True}). Therefore, when breaking a virtual node that has hybrid cycles, we need to first check those odd cycles to make sure that these cycles are satisfiable. An odd cycle will make the world satisfiable if and only if it has \textit{True} nodes. In a virtual node, the only way for an odd cycle to have \textit{True} node would be by overlapping with an even cycle. The overlapping node is assigned \textit{True} which means that the even cycle admits only possible world, that is the one with the overlapped node as \textit{True}.

When there is an odd cycle which does not overlaps with any even cycle, no assignment is possible. If there is neither even cycle, nor odd cycle, the only possible situation is that all nodes that are connected by positive edges. Since we are in the virtual node which has no predecessors, there is no way to make these nodes \textit{True}; so we assign \textit{False} to every node, and delete all edges.

\medskip\noindent\textbf{Performance} \label{sec:implementation} 
The grASP system has been written in python and uses the $DiGraph$ data structure and \textit{simple\_cycles} function from NetworkX ~\cite{SciPyProceedings_11}. The performance testing for grASP was done on two types of programs: (i) established benchmarks such as N-queens; (ii) randomly generated answer set programs. Clingo \cite{DBLP:journals/corr/GebserKKS14} was chosen as the system to compare with. For the first phase, we chose four classic NP problems (map coloring problem, Hamiltonian cycle problem, etc.). The results are shown in Table \ref{tb1}. For the second phase, we used a novel propositional ASP program generator that we have developed for this purpose to generate random programs. The testing performed five rounds with 100 programs each (Table \ref{tb2}). The performance comparison shows that for programs with simpler cycle conditions, grASP achieved similar speed to Clingo, but when solving programs with large number of cycles, grASP is slowed down by the cycle breaking process. 

More details about grASP can be found at arXiv \cite{grASP2021fang}.

\begin{table}
\parbox{.5\linewidth}{
\centering
\scalebox{0.9}{
\begin{tabular}{lrr}
\toprule
Problem  & Clingo & grASP \\
\midrule
Coloring-10 nodes & 0.004  & 0.693\\
Coloring-4 nodes & 0.001  & 0.068\\
Ham Cycle-4 nodes(no cross edges) & 0.001  & 0.052\\
Ham Cycle-4 nodes(fully connected) & 0.002  & 0.089\\
Birds    & 0.001  & 0.001\\
Stream Reasoning   & 0.001  & 0.001\\
\bottomrule
\end{tabular}
}
\caption{Performance Comparison on Classic Problems}
\label{tb1}
}
\parbox{.5\linewidth}{
\centering
\scalebox{0.9}{
\begin{tabular}{lrrrrr}
\toprule
Round & \#Rules & \#EC & \#OC  & Clingo & grASP \\
\midrule
1 & 3231 & 61813 & 61781 & 0.032  & 0.351\\
2 & 3078 & 28593 & 29040 & 0.027  & 0.276\\
3 & 3307 & 39346 & 39433 & 0.028  & 0.087\\
4 & 3069 & 14581 & 14868 & 0.022  & 0.405\\
5 & 3074 & 23017 & 22984 & 0.024  & 0.801\\
\bottomrule
\end{tabular}
}
\caption{Performance  on Random Problems (time in seconds, EC: even cycles, OC: odd cycles)}
\label{tb2}
}

\end{table}

\section{igASP: A Top-down Approach}

Our top-down approach is called igASP, which stands for incremental graph-based ASP solver. The philosophy of igASP is to translate an ASP program into a CNR dependency graph, which is always constrained by some constraint rules, then try to satisfy the constraints by assigning presumed truth values to the related nodes, until all constraints have been satisfied. At the same time, igASP will propagate truth values of the nodes whose truth value has already been determined.

Our graph-based approach performs reasoning in an incremental manner. It starts from the constraints in the answer set program and traces along  causal nodes until it find support through facts (well-founded case) or it detects a cycle through negation (cyclic case). Our algorithm can be thought of as a more general form of the Galliwasp algorithm for query-driven execution of answer set programs \cite{galliwasp2}. The igASP approach is constraint-driven and thus significantly reduces the search space by avoiding exploration of worlds that are inconsistent with the constraints. Furthermore, the incremental reasoning from constraints allows igASP to perform query-driven execution.

\subsection{The igASP Algorithm} \label{sec:igASPalgorithm}

  The igASP algorithm is a recursive algorithm. Since a CNR dependency graph represents the causal relationships among nodes, a topological order would indicate the truth values flow along edges from one node to another starting from the leaves. By their nature, the constraint nodes (labeled \textit{False}, discussed in Section \ref{sec:dg}) will be at the end of such flows in the CNR dependency graph. Therefore, we can incrementally establish the satisfiability relationships across all the nodes starting from the constraint nodes. This incremental establishment of satisfiability starting from the constraint nodes amounts to developing a proof tree. An example (Program 5) is shown in Figure \ref{fig:incre-graph} where the graph is to the left and proof tree to the right. To falsify the constraint node, i.e., to ensure it is \textit{False}, node $m$ must be \textit{True} and $n$ must be \textit{False}. For node $m$ to be \textit{True}, at least one of the three must hold: $p$ is \textit{True}, $q$ is \textit{False}, or $r$ is \textit{True}. When every node's presumed truth value has been found to be consistent with all the dependencies, the algorithm will return the answers.
\begin{figure}[tb]
    \centering
    \includegraphics[scale=0.3]{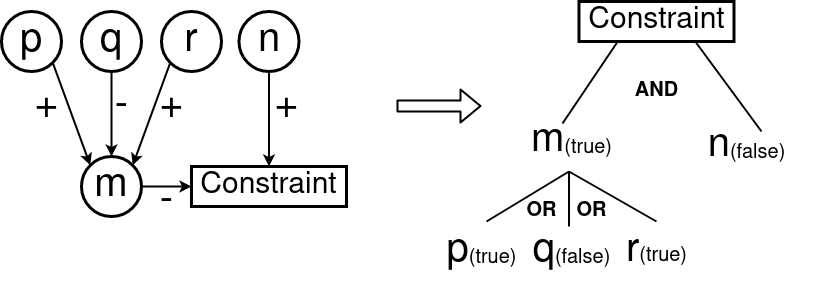}
    \caption{Satisfiability Example}
    \label{fig:incre-graph}
\end{figure}
\begin{lstlisting}[language=prolog,basicstyle=\small]
    %% program 5
    m :- p. m :- not q. m :- r. :- not m. :- n.
\end{lstlisting}

\medskip\noindent\textbf{Effective Edge:}
An effective edge in a CNR dependency graph refers to any edge that propagates \textit{True} value to the node it is incident on. There are two type of \textit{effective} edges: (i) positive edge emanating from a \textit{True} node; (ii) negative edge emanating from a \textit{False} node. An effective edge only points to a \textit{True} node.

\medskip\noindent\textbf{Satisfying Conjunction vs. Disjunction:}
There are two kinds of dependencies that may arise for a node in the CNR dependency graph: conjunctive and disjunctive. A conjunctive dependency refers to the situation where a node is presumed to be \textit{False}. In this case, none of the edges incident into the node should be effective edges. A disjunctive dependency indicates that when a node is presumed to be \textit{True}, at least one of the in-coming edges should be effective. Since igASP works in a reverse manner (from constraints to facts), we may get multiple partial models before we can validate the \textit{True}/\textit{False} label of the current node. For both conjunction and disjunction, these partial models need to be merged for the sake of integrity as well as efficiency. The merging process is discussed later.

\medskip\noindent\textbf{Proof Branch:}
In the igASP algorithm, we start from the constraints that have to be shown to be false, and incrementally construct a proof tree obeying the constraints imposed by the CNR dependency graph. In this incremental reasoning  process, we will pursue various paths in the CNR dependency graph. Our proof will have multiple branches, corresponding to various paths in the CNR dependency graph. Traversal of a branch stops when we reach a fact node whose value has already been given by the ASP program (i.e., known to be a true due to being a fact or known to be false because the atom does not have a rule with that atom as head), or sense that the branch contains a \textbf{cycle}.

\medskip\noindent\textbf{Cycle Handling:}
For positive cycles, we need to ensure that any models computed are consistent with ASP semantics. Suppose we have a program {\tt p :- q. q :- p.}, under ASP semantics, it will have only one answer set:  \{p/False, q/False\}. The other model (\{p/True, q/True\}) has to be rejected, as it is not well-founded per ASP semantics. Thus, positive cycles have to be handled properly so that only correct answer sets are reported. 

To detect a cycle, igASP keeps track of the presumed nodes along the branch, when the current node has been seen previously, we will check whether there exist any \textit{False} node between these two nodes. If so, it is an even cycle, otherwise, it is a positive cycle and only the falsifying assignment should be computed. 

\medskip\noindent\textbf{Model Merging}
As mentioned previously, for each presumed node $n$ (i.e., a node assigned a truth value), its dependencies will be either conjunctive (if $n$ is presumed \textit{False}) 
or disjunctive.
(if $n$ is presumed \textit{True}).
For both conditions, we need to merge the partial models that have been computed so far while assigning a truth value to the dependent nodes.
For the conjunctive condition, the merging process only takes successfully merged models, each of which are the union of two non-conflicting sub-models. 
For example, consider a node $n$ that is presumed to be \textit{False}. Suppose it has two predecessors $p$ and $q$, both $p$ and $q$ connect to $n$ via negative edges. So that $n$ will only be \textit{False} when both $p$ and $q$ will be \textit{True}. We need a conjunctive merge here. Suppose we have sub-models \{p1:\{a/True, d/True, b/False\}, p2:\{a/False, b/True\}\} that hold for $p$ to be \textit{True}, and sub-models \{q1:\{a/True, c/True, b/False\}\} for $q$ to be \textit{True}. The conjunction merging of sub-models between $p$ and $q$ will only accept the union of $p1$ and $q1$, because $p2$ conflicts with $q1$. Therefore, there will only be one model to satisfy for $n$ being \textit{False}, that is \{a/True, c/True, d/True, b/False\}.

For a disjunctive merging, we will keep the conflicted sub-models along with successfully merged ones. Let's modify the above example a little bit by presuming the value of node $n$ to be \textit{True}, and keep everything else unchanged. Now the merging condition became disjunctive, because one of $p$ or $q$ being \textit{False} will still make $n$ \textit{True}. Since $p1$ and $q1$ can be merged without conflict, we replace them by their union \{a/True, c/True, d/True, b/False\}. But this time we don't discard $p2$, because $p2$ is also a valid model that makes $n$ \textit{True}. Therefore, after this merging, we will have two sub-models for $n$ being \textit{True}: \{\{a/True, c/True, d/True, b/False\}, \{a/False, b/True\}\}.

\medskip\noindent\textbf{Forward Propagation:}
Since nodes are assigned values is in a backward chaining manner, where we compute the truth assignment of the predecessors before that of the current node, the sub-models needs to cover as much information as possible. If some nodes' value can be inferred from the proven nodes, they must also be added into the sub-model. For example, suppose we have a sub-model \{a/True, b/False\} for making node $n$ \textit{True}. Suppose there are two additional rules related to node $a$ and $b$: {\tt (i) c :- a. (ii) d :- not b.} In this case, we know that $c$ and $d$ must also be \textit{True}. 

igASP propagates truth values every time a presumed node value has been established, by using a causal map which covers all of the causal relationships for each node/value. When a presumed node/value is established, igASP will check whether there is any other node whose value can be inferred from current node assignments. If there are any, the inferred value is assigned to that node and propagation continues until the model does not change.

\medskip\noindent\textbf{Query Handling:}
A query w.r.t. an ASP program amounts to checking whether a literal is in one of the models of the program. For instance, ASP program {\tt p :- not q. q :- not p. :- p, q.} has two models \{\{p/True, q/False\}, \{p/False, q/True\}\}. If we query $p$, we should get the model $\{p\}$. 

For query handling, igASP negates the query literal and append it to the ASP program as an additional constraint. So for the above example, the query $p/True$ will be converted to a constraint rule {\tt :- not p.} and added to the original program. So the program will now be {\tt p :- not q. q :- not p. :- p, q. :- not p.}

\medskip\noindent\textbf{Non-constrained (Non-headless-rules) Program Handling:}
igASP begins its reasoning from a constraint node (typically, the query represented as a constraint), then searches for a partial answer set to satisfy the constraint. 

This may raise a concern: How about an ASP program that has no global constraints (headless rules) at all? To solve this problem, igASP performs a conversion on the original dependency graph.

Since all original facts in an ASP program should never be \textit{False}. It means that we can take all negated facts as global constraints. Therefore, for any ASP program has default facts, we will generate global constraints accordingly. What if there is no fact in the program? In this case, igASP picks one node, and links it to the ``Constraint" node with both positive and negative edges (via a conjunction node). The reason is simple, a node will either be \textit{True} or \textit{False}. For a program whose dependency graph is disconnected, igASP picks one node from each separated sub-graph, and links them to the ``Constraint" node with both positive and negative edges. For picking which node to connect with the ``Constraint" node, we use a heuristic which chooses the node with most in-coming edges. Since in-coming edges represent dependencies, and each sub-graph is connected, the heuristic is admissible. 

\subsection{A Partial Implementation}
Currently, igASP is still under implementation, we are exploring to apply CDCL (Conflict driven clause learning) to improve the performance. Meanwhile, we have developed an application based on a partial implementation of igASP, which only finds partial answer sets. In some application scenarios, finding the whole answer sets may be an overkill or even inapplicable. Especially for commonsense reasoning, where the knowledgebases may be large and guaranteeing consistency may be hard as parts of the knowledgebase were constructed separately. If a subset of the knoweldgebase that contains the answer we are seeking is consistent, we may not care about other inconsistent part of the knowledgebase. In this case, a partial answer set solver will be preferred. This application system (DiscASP) will be presented at the Technical Communications of ICLP 2021.

\section{Causal Justification} \label{sec:causaljustification}

A major advantage of our graph approaches is that they provides justification as to why a literal is in an answer set for free. Providing justification is a major problem for implementations of ASP that are based on SAT solvers. In contrast to SAT-based ASP solvers, our graph representation maintains the information about structure of an ASP program while computing stable models. Indeed, the resulting graph itself is a justification tree. Since the truth values of all vertices are propagated along edges, we are able to find a justification by looking at the effective out-going edges and their ending nodes. Here the effective out-going edge refer to an edge that actually propagated \textit{True} value to its ending node. According to propagation rules that are discussed in Section \ref{sec:propagate}, there are only two type of \textit{effective} out-going edges: (i) positive edge coming from a \textit{True} node; (ii) negative edge coming from a \textit{False} node. Every effective out-going edge should point to a \textit{True} node. Therefore, the justification first picks effective out-going edges, then check each edge's ending node. If all those ending nodes are \textit{True}, the answer set is justified. 

\section{Current Status and Future Works}
We proposed two dependency graph based approaches to compute the answer sets of an answer set program. We use a novel transformation to ensure that each program has a unique dependency graph, as otherwise multiple programs can have the same dependency graph. A major advantage of our algorithm is that it can produce a justification for any proposition that is entailed by a program. 

Currently, grASP has been finished as the first working version, while igASP is still under implementation. Meanwhile, we have implemented an application of the partial version of igASP, which compute the atoms of an answer set that are related to a query atom within a fixed ``causal distance". This application system will be published at the Technical Communication of ICLP 2021. 

For now, both approaches only work for propositional answer set programs. Our goal is to extend it so that answer sets of datalog programs (i.e., answer set programs with predicates whose arguments are limited to variables and constants) can also be computed without having to ground them first. This will be achieved by dynamically propagating bindings along the edges connecting the nodes in our algorithm's propagation phase. 

Even though the speed of execution on grASP is slower compared to Clingo, it still finds solutions to NP-hard problems in a reasonable time. We expect that igASP will be much more efficient due to its constraint-driven nature. We plan to investigate optimizing techniques such as conflict driven clause learning \cite{silva2003grasp,gebser22} to speed up execution. Our graph approaches are more than just ASP solvers: their visualization feature makes it suitable for educational purpose and for debugging. Moreover, graph-based approaches bring new possibilities for applying optimization. 

\bibliographystyle{eptcs}
\bibliography{generic}
\end{document}